\documentclass[10pt,conference]{IEEEtran}
\usepackage[cmex10]{amsmath}
\usepackage[short]{optidef}
\usepackage{color}
\usepackage{commath}
\usepackage{algorithm}
\usepackage{algorithmic}
\usepackage{textcomp}
\usepackage{gensymb}
\usepackage{multirow}
\usepackage{multicol}
\usepackage{mathtools}

\ifCLASSINFOpdf
   \usepackage[pdftex]{graphicx}
\graphicspath{ {figures/} }
\else
\fi
\usepackage{graphicx}
\usepackage[font=footnotesize,caption=false,subrefformat=parens,labelformat=parens]{subfig}

\usepackage{cite}
\usepackage{amssymb}
\usepackage{amsthm}
\usepackage{pifont}

\usepackage{siunitx}

\usepackage{eucal}
\usepackage{enumitem}
\newtheorem{prop}{Lemma}
\newtheorem{thm}{Theorem}
\newtheorem{cor}{Remark}


\newcommand\blfootnote[1]{%
		\begingroup
		\renewcommand\thefootnote{}\footnote{#1}%
		\addtocounter{footnote}{-1}%
		\endgroup
}

\title{Handover-Count based Velocity Estimation of Cellular-Connected UAVs}
\author{
   \IEEEauthorblockN{Md Moin Uddin Chowdhury, Priyanka Sinha, and Ismail G{\"{u}}ven{\c{c}}}
    \IEEEauthorblockA{Department of Electrical and Computer Engineering, North Carolina State University, Raleigh, NC, USA\\
        \{mchowdh,psinha2,iguvenc\}@ncsu.edu}
}
\begin{document}
\pdfoutput=1
\maketitle
\begin{abstract}
Cellular-connected unmanned aerial vehicles (UAVs) are expected to play a major role in various civilian and commercial applications in the future. While existing cellular networks can provide wireless coverage to UAV user equipment (UE), such legacy networks are optimized for ground users which makes it challenging to provide reliable connectivity to aerial UEs.  
To ensure reliable and effective mobility management for aerial UEs, estimating the velocity of cellular-connected UAVs carries critical importance. In this paper, we introduce an approximate probability mass function (PMF) of handover count (HOC) for different UAV velocities and different ground base station (GBS) densities. Afterward, we derive the Cramer-Rao lower bound (CRLB) for the velocity estimate of a UAV, and also provide a simple unbiased estimator for the UAV's velocity which depends on the GBS density and HOC measurement time. Our simulation results show that the accuracy of velocity estimation increases with the GBS density and HOC measurement window. Moreover, the velocity of commercially available UAVs can be estimated efficiently with reasonable accuracy. 


\end{abstract}
\section{Introduction}
\blfootnote{This work is supported by NSF grant CNS-1453678.}
Thanks to their flexibility in deployment as well as low production cost, using unmanned aerial vehicles (UAVs) or drones for a wide range of commercial and civilian applications have gained significant interest in recent years~\cite{rui1,geraci_2018,halim_mobility}. For taking the full advantage of UAV deployment, beyond visual line of sight (BVLOS) operations are of critical importance where UAVs can fly autonomously without direct human control. Existing cellular networks can be a strong candidate for deploying autonomous UAVs in BVLOS scenarios, in which the UAVs act as aerial users and can maintain reliable communication for safety and control purposes with the ground base stations (GBSs) in the downlink~\cite{geraci_2018}. For maintaining reliable and good connection quality at the cellular-connected UAVs, effective mobility management (MM) by minimizing handover (HO) failures, radio link failure, as well as unnecessary HOs is critically important. However, due to being served by sidelobes of the GBS that provide lower antenna gains, a UAV might be connected with a GBS located far from it~\cite{geraci_2018,ramy_antenna}. This phenomenon, in turn, makes the reference signal received power (RSRP) based MM of cellular-connected UAVs extremely challenging.

Velocity estimation of a cellular-connected UAV can play an important role in effective MM. This information, in turn, can help efficient resource scheduling, load balancing, and energy efficiency enhancements~\cite{arvind_ho}. Especially, due to the patchy signal coverage of GBS in the sky, a high UAV velocity indicates that the UAV in interest will be associated with a GBS for a brief amount of time. Moreover, UAVs flying at high altitudes suffer from high interference stemming from nearby GBS due to the near free-space path-loss trend in the GBS-to-UAV link~\cite{lin_sky}. Estimating the mobile UAV velocity will enable the GBSs to coordinate among themselves for leveraging the inter-cell interference coordination scheme as done for ground users in~\cite{david2012}. While the global positioning system (GPS) can be used to estimate velocity, it is not a practical solution for power-limited cellular-connected UAVs, since GPS receivers consume a significant amount of power. 

Existing cellular networks can estimate the mobility state of a user into three classes: low, medium, and high mobility~\cite{david2012}. In \cite{arvind_ho}, authors presented approximate probability mass functions (PMFs) for HOC and based on them proposed an efficient estimator of ground user velocity. Using tools from stochastic geometry, analytical studies for HO-rate in typical cellular networks are conducted in~\cite{lin_mobility}, while authors in~\cite{bao_mobility} also considered the presence of small base stations along with GBSs. However, none of these prior works took the MM of aerial users into account.

Research directions in integrating UAVs into existing cellular networks as aerial user equipment (UE) have recently attracted substantial attention. For instance, real-world experiments were conducted to test the feasibility of integrating UAVs as UE in \cite{lin_sky,lin_field,denmark_uav_test}. The Third Generation Partnership Project (3GPP) also studied the challenges in providing reliable UAV mobility support in~\cite{3gpp}. In~\cite{Ramy_walid_ICC2020}, the authors explored the effects of practical antenna configurations on the MM of cellular-connected UAVs. By leveraging machine learning algorithms, works in \cite{martins2019,lin_rogue_2019}, studied the problem of detecting UAVs based on radio signals. However, none of these prior works considered the problem of estimating the velocity of cellular-connected UAVs. To the best of our knowledge, this is the first attempt to estimate UAV velocity in a realistic cellular network based on handover count (HOC) statistics.

Our main contribution in this paper is a novel and efficient HOC based UAV velocity estimation technique while considering realistic GBS antenna radiation pattern~\cite{rebato_antenna} and HO scenario~\cite{3gpp.36.331}, and 3GPP specified GBS-to-UAV path loss model~\cite{3gpp}. We also consider the presence of correlated shadowing~\cite{gudmundson_crr_shad} on the UAV trajectory. Through extensive Matlab simulation, we approximate the HOC probability mass function (PMF) using Poisson distribution and then by using the Matlab \textit{curve fitting toolbox}, we express the Poisson PMF parameter with respect to (w.r.t.) different UAV velocities and GBS densities. Using the approximated PMF, an expression for the  Cramer-Rao lower bound (CRLB) of the estimated velocity is derived. Moreover, we provide an efficient minimum variance unbiased (MVU) velocity estimator whose accuracy coincides with the CRLB. Finally, we investigate the accuracy of the estimator for various UAV velocities, GBS densities, and HOC measurement time intervals. 

The rest of this paper is organized as follows. Section~\ref{sec:sys} describes the system model for HO PMF calculation. The approximation of the HOC PMF is presented in Section~\ref{sec:stat}. We derive the CRLB for UAV velocity estimation and provide an efficient unbiased velocity estimator in Section~\ref{sec:crlb}. Simulation results are presented in Section \ref{sec:simulation}. Finally, Section~\ref{sec:Conc} concludes this paper.

\section{System Model}
\label{sec:sys}
\subsection{Network Model}
Let us consider a cellular network in which a single UAV acting as an aerial user, is flying along a two dimensional (2D) linear trajectory (for instance, through the horizontal X-axis) at a fixed height $h_{\textrm{UAV}}$ and velocity $v$. We consider the linear mobility model due to its simplicity and suitability for UAVs flying in the sky with virtually no obstacle. While flying, we assume that the network can track the number of HOs $H$ made by the UAV during a measurement time window $T$. We denote the distance travelled during this measurement duration as $d=vT$. We present the HO procedure later in this Section.

The underlying cellular network consists of GBSs that are deployed with homogeneous Poisson point process (HPPP) $\Phi$ of intensity $\lambda_{\text{GBS}}~ \text{GBSs/km}^2$, and all GBSs have similar height $h_{\textrm{GBS}}$ and transmission power $P_{\textrm{GBS}}$. 
Each GBS consists of three sectors separated by $120$\degree,~while each sector is equipped with $8 \times 1$ cross-polarized antennas downtilted by $\theta_D$. The radiation pattern of each single cross polarized antenna element consists of both horizontal and vertical radiation patterns and these radiation patterns $A_\text{{E,H}}(\phi)$ and $A_\text{{E,V}}(\theta)$ are obtained as \cite{rebato_antenna},\cite{geraci_2018}:
 \begin{align}
   A_{\mathrm{E,H}}(\phi) &= - \text{min}\left\{ 12\left(\frac{\phi}{\phi_{3\mathrm{dB}}} \right)^2, \mathrm{A_m}\right\}, \\
   A_{\mathrm{E,V}}(\theta) &= - \text{min}\left\{ 12\left(\frac{\theta-90}{\theta_{3\mathrm{dB}}} \right)^2, \mathrm{SLA_V}\right\},
\end{align}
where $\phi \in [0^\circ,360^\circ]$,  $\theta \in [0^\circ,180^\circ]$, $\phi_{3\mathrm{dB}}$ and $\theta_{3\mathrm{dB}}$ are $3$ dB beamwidth with similar value of $65\degree$, $\mathrm{A_m}$ and $\mathrm{SLA_V}$ are front-back ratio and side-lobe level limit, respectively, with identical value of 30 dB. Then the 3D antenna element radiation pattern for each pair of $(\theta,\phi)$ can be expressed as:
\begin{equation}
A_{\mathrm{E}}(\theta,\phi)= G_{\mathrm{max}} - \text{min}\left\{-[A_{\mathrm{E,H}}(\phi)+A_{\mathrm{E,V}}(\theta)],\mathrm{A_m}\right\}.
\label{eq_element}
\end{equation}\par

The array radiation pattern with a given element radiation pattern from \eqref{eq_element} can be calculated as, $A_{\mathrm{A}}(\theta,\phi)=  A_{\mathrm{E}}(\theta,\phi)+\text{AF}(\theta,\phi,n).$ The term $\text{AF}(\theta,\phi,n)$ is the array factor with the number $n$ of antenna elements, given as:
\begin{equation}
\text{AF}(\theta,\phi,n)=10\log_{10}\big [1+\rho \big(|\mathbf{a}~.~\mathbf{w}^T|^2-1 \big)\big],
\end{equation}
where $\rho$ is the correlation coefficient, set to unity. The term $\mathbf{a}$ $\in \mathbb{C}^{n}$ is the amplitude vector set as 
$1/\sqrt{n}$. The term $\mathbf{w}$ $\in \mathbb{C}^{n}$ is the beamforming vector, which can be expressed as:
\begin{equation}
  \mathbf{w} = [w_{1,1}, w_{1,2}, . . . , w _{m_V,m_H}],  
\end{equation}
where~ $m_Vm_H=n$, $w_{p,r}=e^{j2\pi \big((p-1)\frac{\Delta V}{\lambda} \Psi_p +(r-1)\frac{\Delta H}{\lambda} \Psi_r \big)}$, \\$\Psi_p = \cos({\theta})-\cos({\theta_D})$, and $\Psi_r = \sin({\theta})\sin({\phi})-\sin({\theta_D})\sin({\phi_D})$. $\Delta V$ and $\Delta H$ stand for the spacing distances between the vertical and horizontal elements of the antenna array, respectively. We consider $\Delta V=\Delta H=\frac{\lambda}{2}$, where $\lambda$ represents the wavelength of carrier frequency $\text{f}_c$. 
 
We assume that the UAV is equipped with an omnidirectional antenna and the UAV is capable of mitigating the Doppler effect~\cite{rui1}. 

\subsection{Path-loss Model}
For modeling the path-loss between a GBS and the UAV, we consider the RMa-AV-LoS channel model specified by 3GPP\cite{3gpp}. The instantaneous path-loss (in dB) under a line-of-sight (LOS) scenario between GBS $m$ and the UAV can be expressed as:
\begin{align}
\xi_{{m,u}}^{\text{LOS}}(t)&=\text{max}\big(23.9-1.8\log_{10}(h_{\textrm{UAV}}),20 \big)\log_{10}(d_{{m,u,t}})\nonumber\\
&+20\log_{10} \bigg(\frac{40\pi \text{f}_c}{3} \bigg)+\chi_{\text{LOS}},
\end{align}
where  $h_{\textrm{UAV}}$ is between 10 m to 300 m and $\text{f}_c$ is the carrier frequency, while $d_{{m,u,t}}$ represents the 3D distance between the UAV and GBS $m$ at time $t$. $\chi_{\text{LOS}}$ represents the correlated shadow fading (SF) associated with LOS scenario~\cite{Goldsmith:2005:WC:993515}. It is worth noting that the probability of LOS is equal to one if the UAV height falls between 40 m and 300 m \cite{3gpp}.

SF is typically modeled as an independent Gaussian random variable with zero mean and standard deviation $\sigma$. According to \cite{3gpp}, $\sigma$ (in dB) can be expressed as $ \sigma= 4.2\exp({-0.0046~{h}_{\text{UAV}}})$.
However, the SF values of consecutive waypoints of a UAV trajectory might have non-trivial correlation due to high probability of LOS in the GBS-to-UAV link. Hence, we consider that SF is a first-order auto-regressive process~\cite{correlated_shadowing}, where the auto-correlation between the SF values at two points separated by distance $\Delta$ is given by \cite{Goldsmith:2005:WC:993515}, $R(\Delta)=\sigma^2  \beta^{\frac{\Delta}{X_c}}$,
where, $\Delta$ is the distance between the two points, $ \beta$ is the correlation coefficient, and $X_c$ is the decorrelation distance \cite{Goldsmith:2005:WC:993515}. Here, we set $X_c=100$ m and $\beta=0.82$~\cite{Goldsmith:2005:WC:993515}. We first generate independent SF values with zero mean and $\sigma=1$ for each waypoint and then use Cholesky factorization to generate the correlated SF values from $R$~\cite{Klingenbrunn_corr_shadow}. Then, the received power at the UAV from GBS $m$ at time instance $t$ can be expressed as: 
\begin{equation}
P_{rx-m}=P_{\text{GBS}}+ A_{\mathrm{A}}(\theta_{m,t},\phi_{m,t},n), -\xi_{{m,u}}^{\text{LOS}}(t)~,
\label{eq:rsrp}
\end{equation}
where $\theta_{m,t}$ and $\phi_{m,t}$ are the elevation and the azimuth angles, respectively, between the UAV and GBS $m$ at time $t$.
\subsection{Handover Procedure}

The UAV will measure the RSRPs from all the adjacent GBSs at subsequent measurement gaps using \eqref{eq:rsrp}. Here, we consider a HO mechanism that involves a HO margin (HOM) parameter, and a time-to-trigger (TTT) parameter, which is a time window which starts after the following HO condition (A3 event \cite{3gpp.36.331}) is
fulfilled:
\begin{equation}
\text{RSRP}_{{j}}>\text{RSRP}_{{i}}+m_{\text{hyst}},
\label{a3_event}
\end{equation}
where $\text{RSRP}_{{j}}$ and $ \text{RSRP}_{{i}}$ are the RSRPs measured from the serving GBS $i$ and target GBS $j$, respectively and $m_{\text{hyst}}$ is the HOM set by the network operator. The UAV does not transmit its measurement report to its current serving GBS before the TTT expires~\cite{karthik2017}.

\begin{figure}[t]
\centering
	\subfloat[$\lambda_{\text{GBS}}=2$]{
			\includegraphics[width=.48\linewidth]{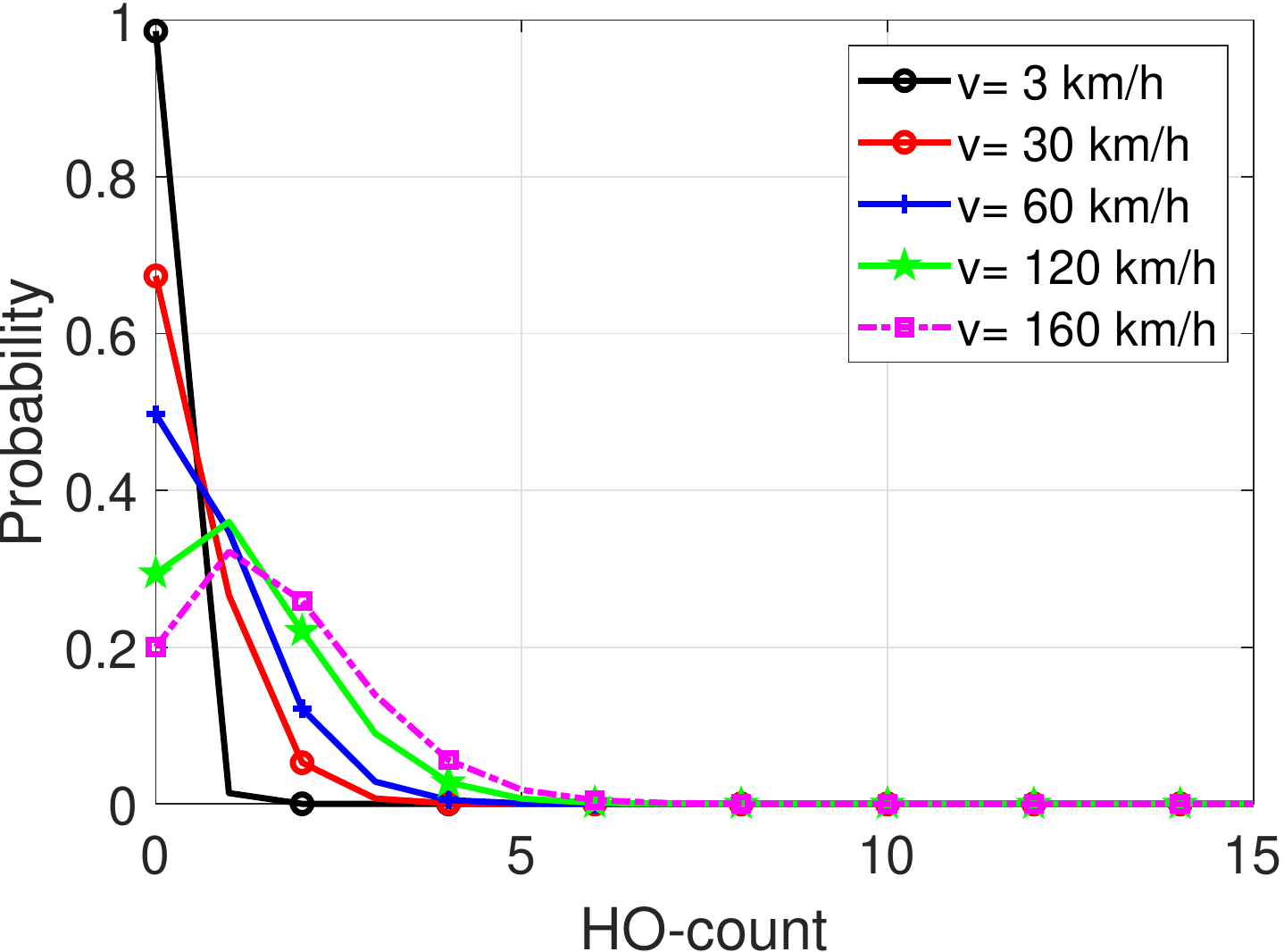}}
		\subfloat[$\lambda_{\text{GBS}}=10$]{
			\includegraphics[width=.48\linewidth]{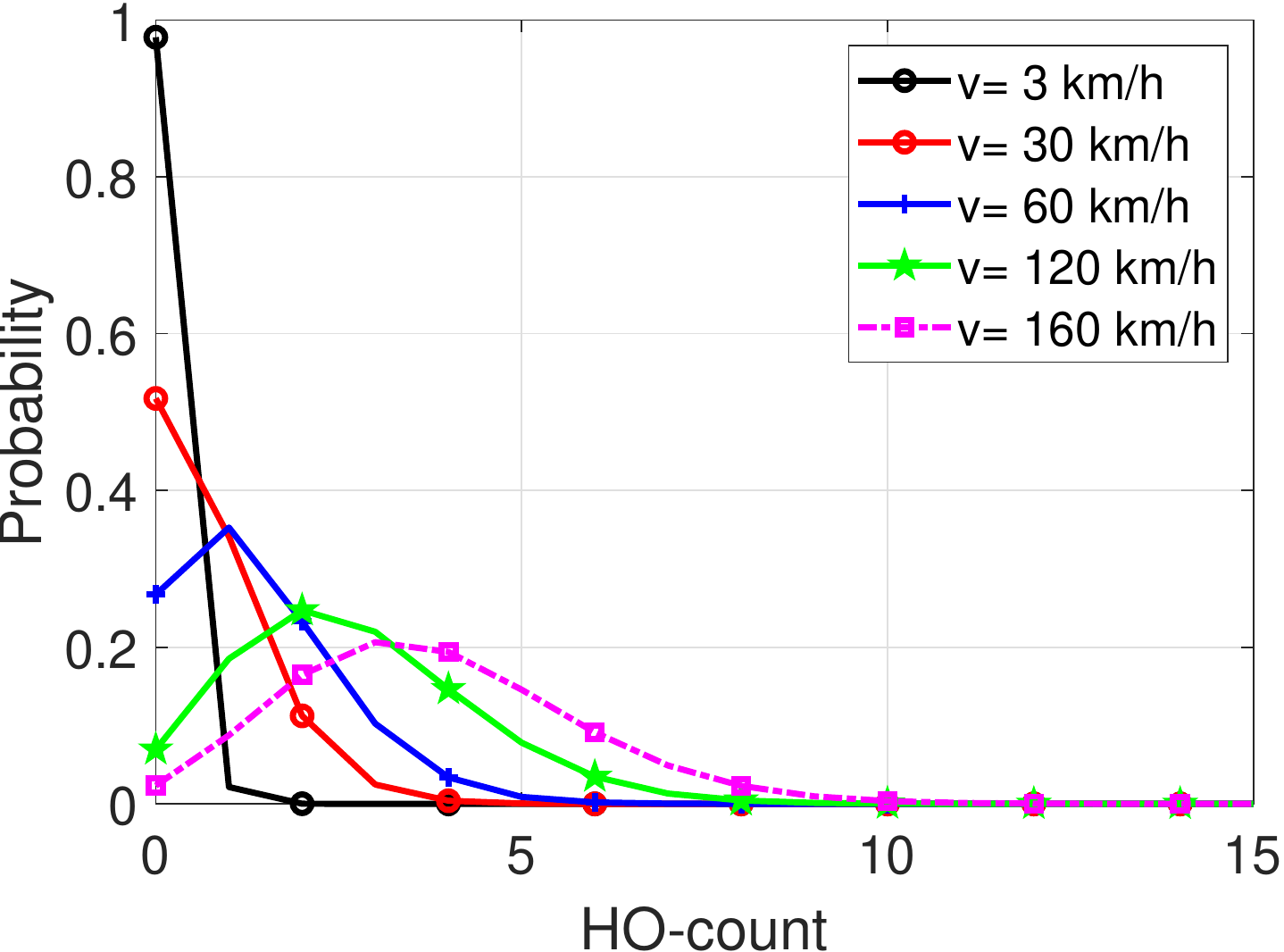}}
    \caption{{PMF of HOC for different $\lambda_{\text{GBS}}$ and $v$ considering SF.}}
    \label{fig:uav_160ms_possion_fit}
    \vspace{-2mm}
\end{figure}

\section{Handover-count Statistics}
\label{sec:stat}
For obtaining the estimated velocity of a UAV based on its HOC, we need to know the HOC PMF $f_H(h)$. To the best of our knowledge, there exists no expression for the PMF of HOC of a cellular-connected UAV. Moreover, due to the intractability in GBS antenna radiation pattern, HO process, and channel models considered in this research, it is extremely difficult to obtain an exact expression of $f_H(h)$.

In Fig.~\ref{fig:uav_160ms_possion_fit}, we plot the HOC PMFs from extensive Matlab simulations for various $v$ and $\lambda_{\text{GBS}}$ values. For each combination of $v$ and $\lambda_{\text{GBS}}$, we obtained 1000 samples of HOC $H$ for constructing the PMF $f_H(h)$. Here, we consider the HOC measurement time interval $T=100$~s. For low values of $\lambda_{\text{GBS}}$, as depicted in Fig.~\ref{fig:uav_160ms_possion_fit}(a), the PMFs for different UAV velocities are overlapping significantly. For higher values of $\lambda_{\text{GBS}}$, the PMFs still overlaps, but they are more spread out which will lead to more accurate velocity estimation. From the obtained HOC data samples we have noticed that the PMFs for different $v$ and $\lambda_{\text{GBS}}$ resembles Poisson distribution.
Hence, we model the PMF of HOC of a UAV flying with a constant velocity with rare parameter $\lambda > 0$ as:
\begin{equation}
f_H^p(h)= \frac{e^{-\lambda}\lambda^h}{h!}.   
\label{eq:PMF}
\end{equation}

For obtaining analytical expression, we express $\lambda$ as a function of distance $d=vT$ and $\lambda_{\text{GBS}}$. Using the MATLAB curve fitting toolbox, we obtain a two-dimensional power fit where the value of
$\lambda$ can be obtained as, $\lambda=a\times\lambda_{\text{GBS}}^b \times (vT)$. For $\text{TTT} = 160$ ms and $m_{\text{hyst}}= 3$ dB, we report the following values of the parameters as, $a=0.2417$ and $b=0.5278$.
In Fig.~\ref{fig:ho-count_fit}, we show the trend of $\lambda$ with respect to (w.r.t) $d$ and $\lambda_{\text{GBS}}$. Since the HOC follows a Poisson distribution, the expected value $E(h)$ and variance $\text{var}(h)$ is $\lambda$.


\begin{figure}[t]
\centering{\includegraphics[width=1\linewidth]{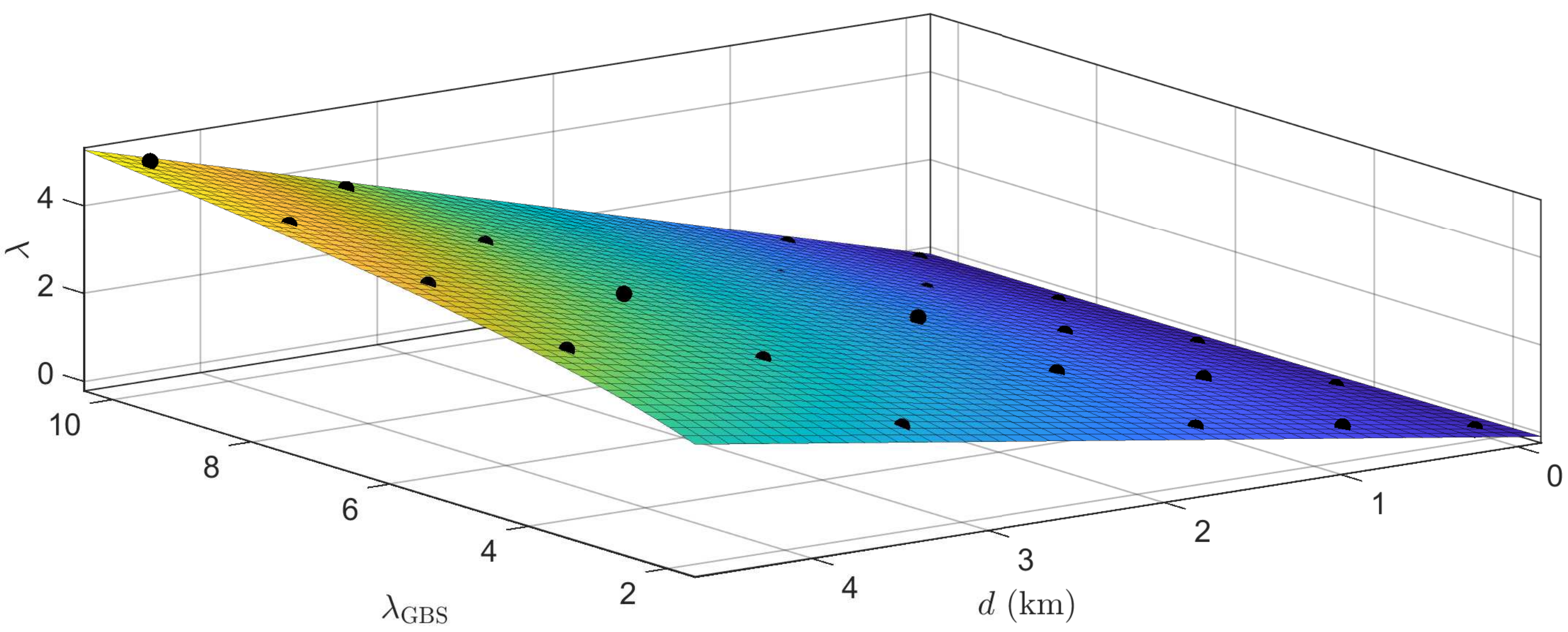}}
\caption{Fitting of the HOC PMF parameter $\lambda$, w.r.t. covered distance $d$ and GBS density $\lambda_{\text{GBS}}$.}
\label{fig:ho-count_fit}
\vspace{-2mm}
\end{figure}

\section{Cramer-Rao Lower Bound for UAV Velocity Estimation}
\label{sec:crlb}
Cramer-Rao lower bound (CRLB) provides a lower bound on the variance of an unbiased estimator. An estimator is considered to be unbiased if the expected value of the estimates coincides with the true value of the parameter of interest. If the variance of an unbiased estimator can achieve the CRLB, it is then said to be an efficient estimator\cite{book:Kay97,arvind_ho}. First, we need to meet the regularity condition~\cite{book:Kay97} for obtaining the CRLB as discussed in the following lemma.

\begin{prop}
The proposed PMF of HOC based on UAV velocity satisfies the regularity condition for obtaining CRLB.
\end{prop}
\begin{proof} We need to show that $E(\frac{\partial \log f_H^p(h;v)}{\partial v})= 0$. First we take the logarithm of $f_H^p(h;v)$ and then differentiate it w.r.t $v$, which can be expressed as follows:  
\begin{eqnarray}
\frac{\partial \log f_H^p(h;v)}{\partial v}= -K \bigg (1-\frac{h}{\lambda}\bigg ),
\label{eq:log_derivative_v}
\end{eqnarray}
where $K=\frac{\lambda}{v}=a\times \lambda_{\text{GBS}}^b \times T$. Finally we take the expectation w.r.t $f_H^p(h;v)$ which can be expressed as:
\begin{align}
E(\frac{\partial \log f_H^p(h;v)}{\partial v}) &= -K+\frac{E(h)}{\lambda}K,\\&=-K+\frac{\lambda}{\lambda}K=0,
\label{eq:regularity}
\end{align}
which completes the proof.
\end{proof}

Next, we present the expression of CRLB by considering the PMF presented in \eqref{eq:PMF}.
\begin{thm}
Let a UAV is flying over a cellular network with GBS density $\lambda_{\text{GBS}}$ per at a fixed height over a linear trajectory and make $H$ handovers at a time period $T$. If the PMF of the HOC can be expressed as $f_H^p(h;v)$ as in \eqref{eq:PMF}, then the CRLB of the estimated velocity is given by,  
\begin{equation}
    \emph{var}(\hat{v}) \geq \frac{v}{K},
    \label{eq:crlb_thm}
\end{equation}
where $K=a\times \lambda_{\text{GBS}}^b \times T$.
\end{thm}
\begin{proof}
By definition of CRLB, we know that 
\begin{equation}
     \text{var}(\hat{v}) \geq \frac{1}{I(v)},
\label{eq:crlb_fisher}
\end{equation}
where $I(v)$ is the \textit{Fisher Information} and can be expressed as:
\begin{equation}
I(v)=\mathop{\mathbb{E}}\Bigg [\bigg(  \frac{\partial \log f_H^p(h;v)}{\partial v}   \bigg )^2 \Bigg].
\end{equation}
Here, $\mathop{\mathbb{E}}[\cdot]$ is the expectation operator w.r.t $H$. By squaring \eqref{eq:log_derivative_v}, and taking the expectation w.r.t $f_H^p(h;v)$, we get:  
\begin{align}
 I(v) &=
\mathop{\mathbb{E}}\Bigg[ K^2(1+\frac{h^2}{\lambda^2} - \frac{2h}{\lambda} ) \Bigg],\\
&= K^2 \bigg(1+\frac{E(h^2)}{\lambda^2} - \frac{2E(h)}{\lambda} \bigg).
\label{eq:CRLB_proof}
\end{align}

Since $h$ follows Poisson distribution with parameter $\lambda$, $E(h^2)$ can be written as $\lambda (1+\lambda)$. Hence, we get: 
\begin{align}
 I(v) &=
K^2\bigg(1+\frac{\lambda (1+\lambda)}{\lambda^2} - \frac{2\lambda}{\lambda} \bigg)= \frac{K^2}{\lambda}=\frac{K}{v}.
\label{eq:CRLB_proof_part2}
\end{align}
By placing the expression of $I(v)$ in \eqref{eq:crlb_fisher}, we can obtain the CRLB for $\hat{v}$ as in \eqref{eq:crlb_thm}.\end{proof}
\begin{cor}
Variance of the estimated velocity $\hat{v}$ increases linearly with the UAV velocity and decreases with increasing HOC duration $T$ and GBS density $\lambda_{\text{GBS}}$. 
\end{cor}

\begin{table}[t]
\centering
\renewcommand{\arraystretch}{1}
\caption {Simulation parameters.}
\label{Tab:Sim_par}
\scalebox{0.99}
{\begin{tabular}{lc}
\hline
Parameter & Value \\
\hline
$P_{\text{GBS}}$ & 46 dBm  \\ 
$\theta_D$ & $6^\circ$ \\ 
$h_{\text{UAV}}$ & 120 m\\ 
$h_{\text{GBS}}$ & 35 m\\ 
${\text{f}_c}$ & 1.5 GHz\\ 
$\lambda$\textsubscript{GBS} & 2, 4, 6, 8, and 10 per $\text{km}^2$\\ 
$v$ & 3 , 30, 60, 120, and 160 kmph\\
$\text{measurement gap}$ & 200 ms\\
TTT, $m_{\text{hyst}}$ & 160~ms and 3 dB\\
\hline
\end{tabular}}
\end{table}

\subsection{Minimum Variance Unbiased Estimator for UAV Velocity}

The derived CRLB for $\hat{v}$ takes the HOC as input in the closed form. In this sub-section, we will derive the variance of this estimator $\hat{v}$ and show that this is indeed an MVU estimator. 
We first consider the Rao-Blackwll-Lehmann-Scheffe (RBLS) theorem to find the MVU velocity estimator \cite[Section~5.5]{book:Kay97}. According to Neyman-Fisher factorization theorem if we can factorize the PMF $f_H^p(h;v)$ as
\begin{equation}
    f_H^p(h;v)=g(F(h),v)r(h),
\end{equation}
where $g$ is function that depends on $h$ only by $F(h)$, then we can conclude that $F(h)$ is a sufficient statistics of $v$\cite[Section~5.4]{book:Kay97}. We can factor the HOC PMF for our case as: 
\begin{equation}
    f_H^p(h;v)=\underbrace{(e^{-\lambda}\lambda^h)}_{g(F(h),v)}~\times~\underbrace{\frac{1}{h!}}_{r(h)}~.
\end{equation}

Hence, $F(h)=h$ is a sufficient statistics of $v$. Since $E(h)=\lambda=Kv$, we can formulate an estimator of $v$ as, $\hat{v}=\frac{h}{K}$. We can calculate the mean of this estimator as:
\begin{equation}
   E(\hat{v})=\frac{E(h)}{K}=\frac{\lambda}{K}=v~,
\end{equation}
which shows that $\hat{v}$ is an unbiased estimator of $v$. Next, we derive the variance of $\hat{v}$ to verify whether this is an efficient estimator as follows: 
\begin{equation}
   \text{var}(\hat{v})=\text{var}(\frac{h}{K})=\frac{\lambda}{K^2}=\frac{v}{K},
\end{equation}
which coincides with the CRLB of $\hat{v}$ as presented in \eqref{eq:crlb_thm} and hence, the derived estimator is an efficient estimator.
\section{Simulation Results}
\label{sec:simulation}

In this section, we first study the accuracy of the Poisson PMF approximation by plotting its mean square error (MSE) performance. Then we study the $\text{var}(\hat{v})$ w.r.t. GBS density $\lambda_{\text{GBS}}$, UAV velocity $v$, and HOC measurement time $T$. Simulation parameters are provided in Table~\ref{Tab:Sim_par}.

The MSE between the approximate PMF and the PMF obtained from simulations can be expressed as:
\begin{equation}
    \text{MSE}=\frac{1}{L}\sum_{l=1}^L \big[ f_H(l)-f_H^p(l)\big]^2,
\end{equation}
where $L$ is the number of samples in the PMFs. In Fig.~\ref{fig:MSE_vs_lambda}, we plot the MSE performance w.r.t. various $v$ and $\lambda_{\text{GBS}}$. We can conclude that our approximated PMF behave very close to the actual PMF obtained from the simulations. 

\begin{figure}[t]
\centering{\includegraphics[width=0.85\columnwidth]{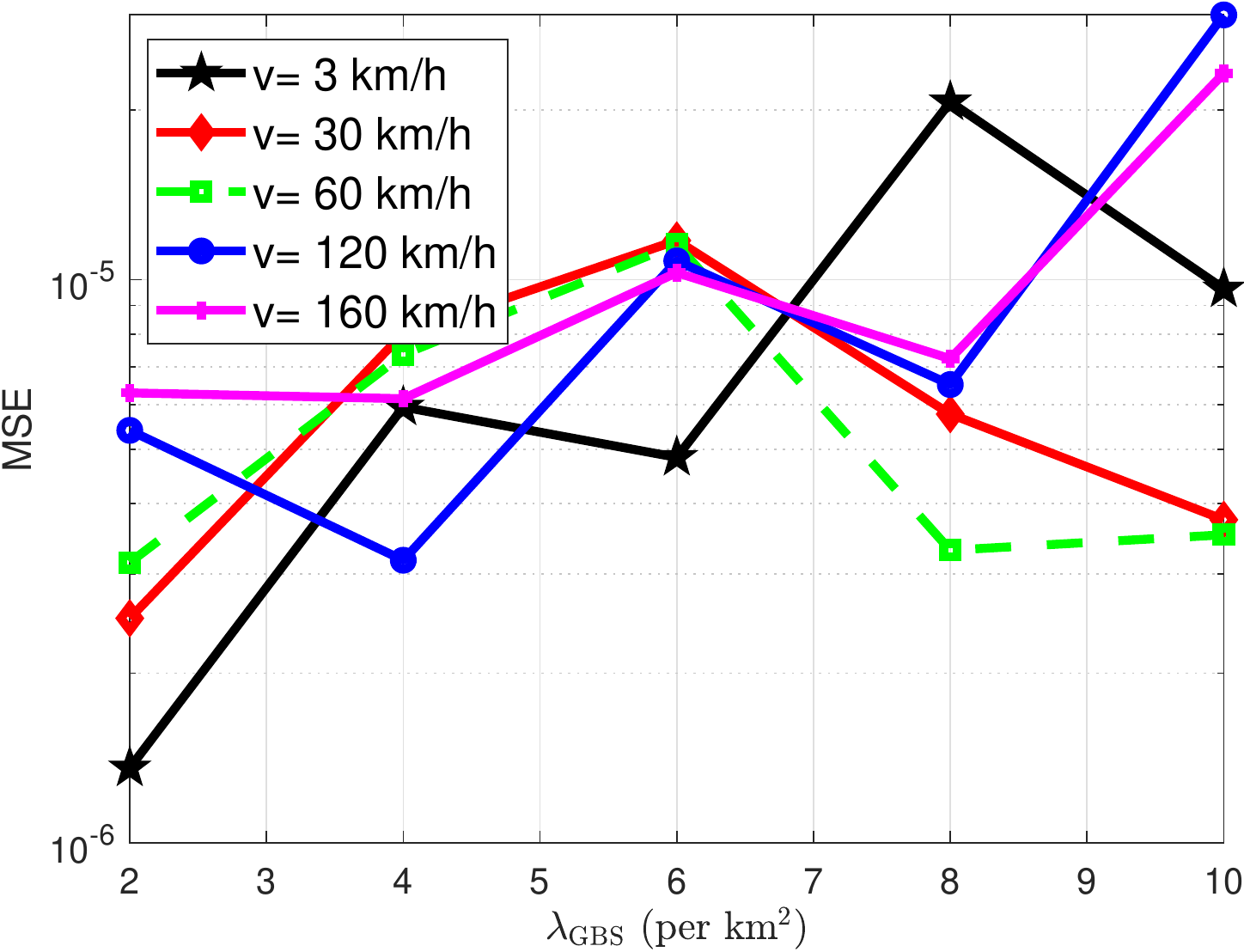}}
    \caption{MSE vs $\lambda_{\text{GBS}}$ for different $v$ and $T=100$~s.}
    \label{fig:MSE_vs_lambda}
   \vspace{-2mm}
\end{figure}

In Fig.~\ref{fig:CRLB_vs_vel}, we plot the square root of the CRLB or standard deviation of the variance of $\hat{v}$ w.r.t.  $\lambda_{\text{GBS}}$. As expected, CRLB of the proposed estimator decreases with increasing $\lambda_{\text{GBS}}$ and increases with $v$ as in \eqref{eq:crlb_thm}. This is in line with Fig.~\ref{fig:uav_160ms_possion_fit}, where the PMFs are closely spaced with each other for low $\lambda_{\text{GBS}}$, making it difficult to distinguish between different $v$. For $T=100$ s and high UAV velocities, our method provides high root mean square errors (RMSEs). However, for available commercial UAVs with maximum velocity of $68$ km/h~\cite{drone_speed}, our proposed simple estimator provides RMSEs of $28$ km/h and $24$ km/h for $\lambda_{\text{GBS}}=6$ and $\lambda_{\text{GBS}}=10$, respectively, with $T=500$~s. For UAV velocities less than $20$ km/h, UAV's velocity can estimated with RMSE less $20$ km/h with $T=500$~s.

\begin{figure}[t]
\centering{\includegraphics[width=0.85\columnwidth]{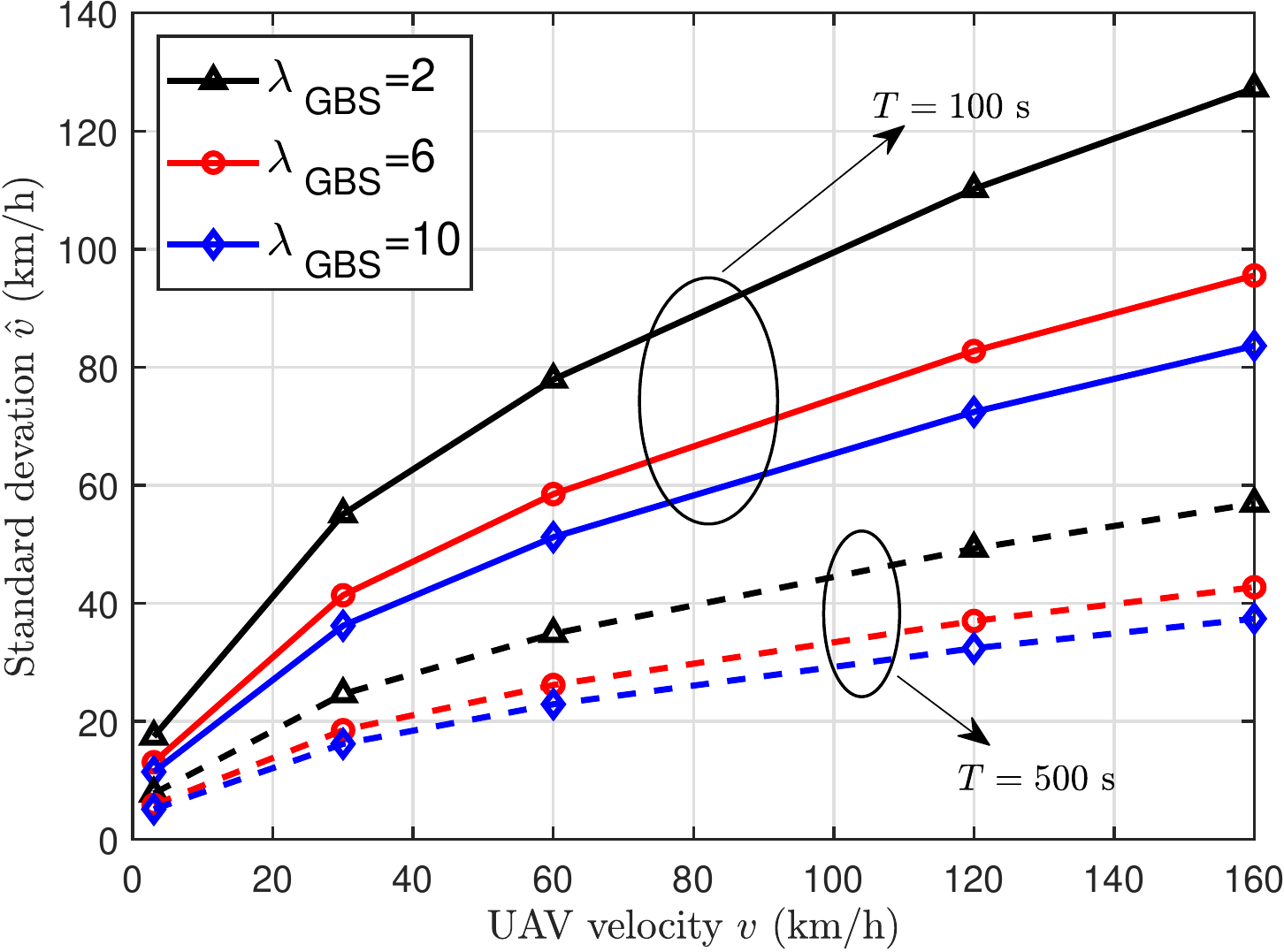}}\vspace{-1mm}
    \caption{CRLB vs $v$ for different $\lambda_{\text{GBS}}$ and $T=100$~s and $T=500$~s.}
    \label{fig:CRLB_vs_vel}
  \vspace{-1mm}
\end{figure}

In Fig.~\ref{fig:CRLB_vs_time}, we show the effect of HOC measurement time $T$ on the CRLB of our proposed estimator. For a given $\lambda_{\text{GBS}}$ and $v$, CRLB decreases with increasing $T$. As expected, higher velocity provides lower accuracy of velocity estimation. Overall, longer HOC measurement window will provide better velocity estimation since HOCs will be more distinguishable for various $v$ if we allow more time to count HOs made by the UAV traveling on a linear trajectory. There exists a trade-off between the rapidness and accuracy of the estimated velocities which is also evident in Fig.~\ref{fig:CRLB_vs_vel}.

\begin{figure}[t]
\centering{\includegraphics[width=0.85\columnwidth]{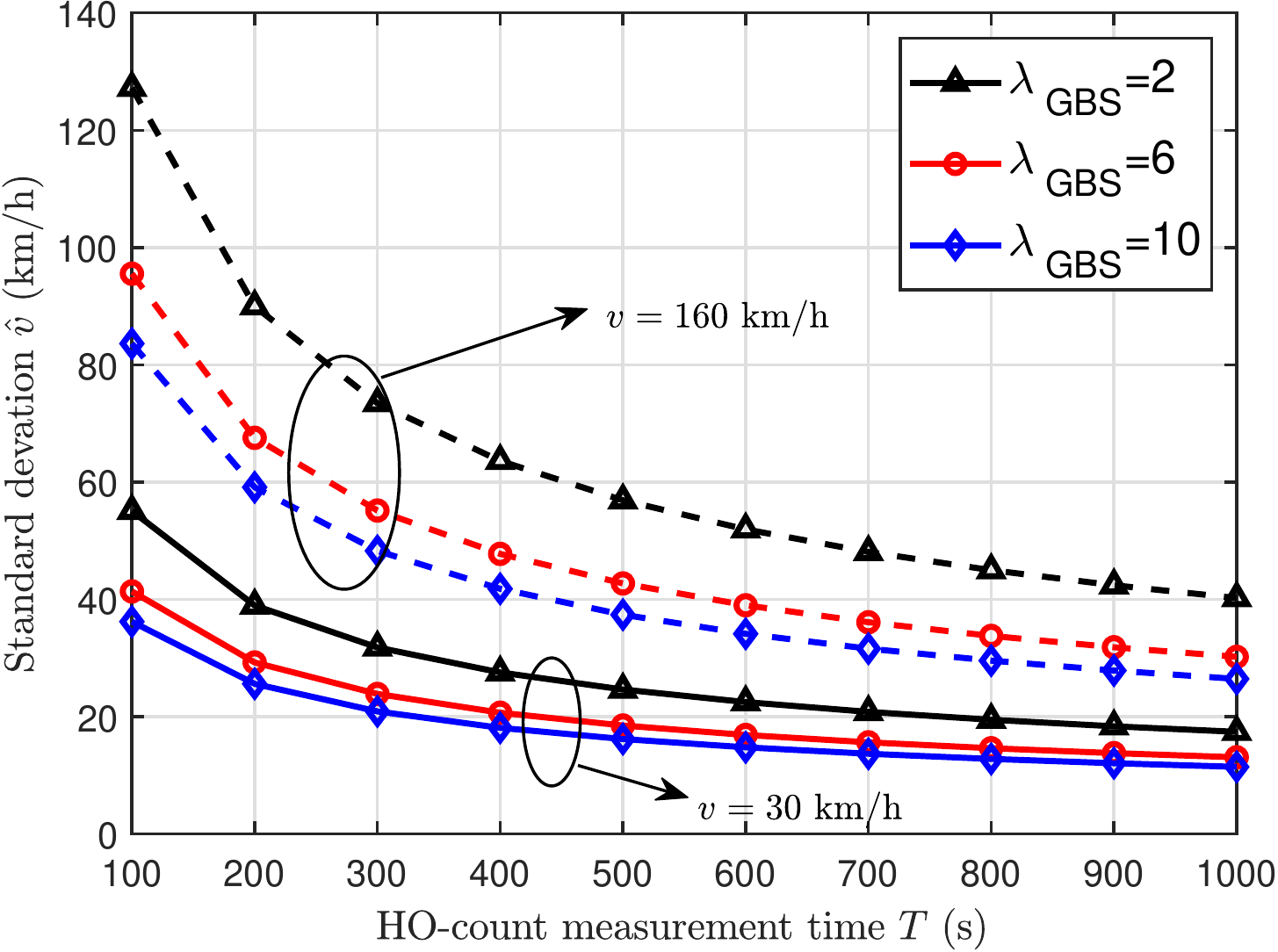}}\vspace{-1mm}
    \caption{CRLB vs $T$ for different $\lambda_{\text{GBS}}$ and $v=30~\text{km}/\text{h}$ and $v=160~\text{km}/\text{h}$.}
    \label{fig:CRLB_vs_time}
    \vspace{-2mm}
\end{figure}

\section{Conclusion}
\label{sec:Conc}
In this paper, we have derived the CRLB for HOC based velocity estimation for a cellular-connected UAV flying with constant height and velocity. We have approximated the HOC PMF using Poisson distribution and showed that the approximated PMF provides a good fit with low MSE. We have also estimated the PMF parameter as a function of GBS density, measurement duration, and UAV velocity, based on which we have proposed a simple UAV velocity estimator. Our results show that the CRLB of the estimated velocity decreases with larger HOC measurement time interval and higher GBS density. Our future work includes estimation of the UAV velocity with three-dimensional flight trajectory.

\bibliographystyle{IEEEtran} 
\bibliography{ref}

\begin{thebibliography}{10}
\providecommand{\url}[1]{#1}
\csname url@samestyle\endcsname
\providecommand{\newblock}{\relax}
\providecommand{\bibinfo}[2]{#2}
\providecommand{\BIBentrySTDinterwordspacing}{\spaceskip=0pt\relax}
\providecommand{\BIBentryALTinterwordstretchfactor}{4}
\providecommand{\BIBentryALTinterwordspacing}{\spaceskip=\fontdimen2\font plus
\BIBentryALTinterwordstretchfactor\fontdimen3\font minus
  \fontdimen4\font\relax}
\providecommand{\BIBforeignlanguage}[2]{{%
\expandafter\ifx\csname l@#1\endcsname\relax
\typeout{** WARNING: IEEEtran.bst: No hyphenation pattern has been}%
\typeout{** loaded for the language `#1'. Using the pattern for}%
\typeout{** the default language instead.}%
\else
\language=\csname l@#1\endcsname
\fi
#2}}
\providecommand{\BIBdecl}{\relax}
\BIBdecl

\bibitem{rui1}
Q.~Wu, Y.~Zeng, and R.~Zhang, ``Joint trajectory and communication design for
  multi-\textsc{UAV} enabled wireless networks,'' \emph{IEEE Trans. Wireless
  Commun.}, vol.~17, no.~3, pp. 2109--2121, Mar. 2018.

\bibitem{geraci_2018}
G.~{Geraci}, A.~{Garcia-Rodriguez}, L.~{Galati Giordano}, D.~{López-Pérez},
  and E.~{Bjornson}, ``Understanding {UAV} cellular communications: From
  existing networks to massive {MIMO},'' \emph{IEEE Access}, vol.~6, pp.
  67\,853--67\,865, 2018.

\bibitem{halim_mobility}
S.~{Enayati}, H.~{Saeedi}, H.~{Pishro-Nik}, and H.~{Yanikomeroglu}, ``Moving
  aerial base station networks: A stochastic geometry analysis and design
  perspective,'' \emph{IEEE Trans. on Wireless Commun.}, vol.~18, no.~6, pp.
  2977--2988, June 2019.

\bibitem{ramy_antenna}
R.~{Amer}, W.~{Saad}, and N.~{Marchetti}, ``Towards a connected sky:
  Performance of beamforming with down-tilted antennas for ground and uav user
  co-existence,'' \emph{IEEE Commun. Lett.}, pp. 1--1, 2019.

\bibitem{arvind_ho}
A.~{Merwaday} and I.~{Guvenc}, ``Handover count based velocity estimation and
  mobility state detection in dense hetnets,'' \emph{IEEE Trans. Wireless
  Commun.}, vol.~15, no.~7, pp. 4673--4688, July 2016.

\bibitem{lin_sky}
X.~{Lin}, V.~{Yajnanarayana}, S.~D. {Muruganathan}, S.~{Gao}, H.~{Asplund},
  H.~{Maattanen}, M.~{Bergstrom}, S.~{Euler}, and Y.~.~E. {Wang}, ``The sky is
  not the limit: {LTE} for unmanned aerial vehicles,'' \emph{IEEE Commun.
  Mag.}, vol.~56, no.~4, pp. 204--210, Apr. 2018.

\bibitem{david2012}
D.~{Lopez-Perez}, I.~{Guvenc}, and X.~{Chu}, ``Mobility management challenges
  in {3GPP} heterogeneous networks,'' \emph{IEEE Commun. Mag.}, vol.~50,
  no.~12, pp. 70--78, Dec. 2012.

\bibitem{lin_mobility}
X.~{Lin}, R.~K. {Ganti}, P.~J. {Fleming}, and J.~G. {Andrews}, ``Towards
  understanding the fundamentals of mobility in cellular networks,'' \emph{IEEE
  Trans. Wireless Commun.}, vol.~12, no.~4, pp. 1686--1698, Apr. 2013.

\bibitem{bao_mobility}
W.~{Bao} and B.~{Liang}, ``Stochastic geometric analysis of user mobility in
  heterogeneous wireless networks,'' \emph{IEEE J. Sel. Areas Commun. (JSAC)},
  vol.~33, no.~10, pp. 2212--2225, Oct 2015.

\bibitem{lin_field}
X.~{Lin}, R.~{Wiren}, S.~{Euler}, A.~{Sadam}, H.~{Maattanen},
  S.~{Muruganathan}, S.~{Gao}, Y.~E. {Wang}, J.~{Kauppi}, Z.~{Zou}, and
  V.~{Yajnanarayana}, ``Mobile network-connected drones: Field trials,
  simulations, and design insights,'' \emph{IEEE Veh. Technol. Mag.}, vol.~14,
  no.~3, pp. 115--125, Sept. 2019.

\bibitem{denmark_uav_test}
H.~C. Nguyen, R.~Amorim, J.~Wigard, I.~Z. Kovács, T.~B. Sørensen, and P.~E.
  Mogensen, ``How to ensure reliable connectivity for aerial vehicles over
  cellular networks,'' \emph{IEEE Access}, vol.~6, pp. 12\,304--12\,317, 2018.

\bibitem{3gpp}
\BIBentryALTinterwordspacing
3GPP, Technical Specification (TS) 36.777, 2018. [Online]. Available:
  \url{https://portal.3gpp.org/desktopmodules/Specifications/\\SpecificationDetails.aspx?specificationId=3231}
\BIBentrySTDinterwordspacing

\bibitem{Ramy_walid_ICC2020}
R.~{Amer}, W.~{Saad}, B.~{Galkin}, and N.~{Marchetti}, ``Performance analysis
  of mobile cellular-connected drones under practical antenna configurations,''
  in \emph{Proc. IEEE Int. Conf. Commun. (ICC)}, Dublin, Ireland, June 2020.

\bibitem{martins2019}
M.~{Ezuma}, F.~{Erden}, C.~{Kumar Anjinappa}, O.~{Ozdemir}, and I.~{Guvenc},
  ``Detection and classification of {UAVs} using {RF} fingerprints in the
  presence of {Wi-Fi} and {B}luetooth interference,'' \emph{IEEE Open Journal
  of the Communications Society}, vol.~1, pp. 60--76, 2020.

\bibitem{lin_rogue_2019}
H.~{Rydén}, S.~B. {Redhwan}, and X.~{Lin}, ``Rogue drone detection: A machine
  learning approach,'' in \emph{Proc. IEEE Wireless Commun. Netw. Conf.
  (WCNC)}, Apr. 2019, pp. 1--6.

\bibitem{rebato_antenna}
M.~{Rebato}, J.~{Park}, P.~{Popovski}, E.~{De Carvalho}, and M.~{Zorzi},
  ``Stochastic geometric coverage analysis in mmwave cellular networks with
  realistic channel and antenna radiation models,'' \emph{IEEE Trans. Commun.},
  vol.~67, no.~5, pp. 3736--3752, May 2019.

\bibitem{3gpp.36.331}
3GPP, ``{Evolved Universal Terrestrial Radio Access (E-UTRA); Radio Resource
  Control (RRC); Protocol specification},'' {3rd Generation Partnership Project
  (3GPP)}, Technical Specification (TS) 36.331, 04 2017, version 14.2.2.

\bibitem{gudmundson_crr_shad}
M.~{Gudmundson}, ``Correlation model for shadow fading in mobile radio
  systems,'' \emph{Electronics Letters}, vol.~27, no.~23, pp. 2145--2146, Nov
  1991.

\bibitem{Goldsmith:2005:WC:993515}
A.~Goldsmith, \emph{Wireless Communications}.\hskip 1em plus 0.5em minus
  0.4em\relax New York, NY, USA: Cambridge University Press, 2005.

\bibitem{correlated_shadowing}
J.~Lee and F.~Baccelli, ``On the effect of shadowing correlation on wireless
  network performance,'' Honolulu, HI, 2018, pp. 1601--1609.

\bibitem{Klingenbrunn_corr_shadow}
T.~{Klingenbrunn} and P.~{Mogensen}, ``Modelling cross-correlated shadowing in
  network simulations,'' in \emph{Proc. IEEE Vehic. Technol. Conf. (VTC)},
  vol.~3, Sep. 1999, pp. 1407--1411.

\bibitem{karthik2017}
K.~{Vasudeva}, M.~{Simsek}, D.~{Lopez-Perez}, and I.~{Guvenc}, ``Analysis of
  handover failures in heterogeneous networks with fading,'' \emph{IEEE Trans.
  Veh. Technol.}, vol.~66, no.~7, pp. 6060--6074, July 2017.

\bibitem{book:Kay97}
S.~M. Kay, \emph{Fundamentals of Statistical Signal Processing: Estimation
  Theory}.\hskip 1em plus 0.5em minus 0.4em\relax Prentice Hall, 1997.

\bibitem{drone_speed}
Dronerush, ``How fast can a drone fly? – how to fly, the science of flight,''
  \url{https://https://dronerush.com/how-fast-can-a-drone-fly-science-of-flight-10953/}.

\end{thebibliography}
\end{document}